\documentclass[twoside, onecolumn]{IEEEtran}
\usepackage[utf8]{inputenc} 
\usepackage[T1]{fontenc}
\usepackage{url}
\usepackage{ifthen}
\usepackage{cite}
\usepackage[cmex10]{amsmath} 

\usepackage[
font=small,labelfont=bf
]{caption}
 \usepackage{flushend}
 \flushend
\usepackage{amsfonts, amssymb, amsthm, color,mathtools}
\usepackage{bbm}
\usepackage{comment,enumitem}

\usepackage{tikz}
\usepackage{diagbox}
\usepackage{caption}
\usepackage{subcaption}
\newcommand{\poly}{{\rm poly}}
\renewcommand{\a}{\boldsymbol{a}}
\renewcommand{\b}{\boldsymbol{b}}
\renewcommand{\c}{\boldsymbol{c}}
\newcommand{\x}{\boldsymbol{x}}
\newcommand{\y}{\boldsymbol{y}}
\renewcommand{\u}{\boldsymbol{u}}
\newcommand{\z}{\boldsymbol{z}}
\newcommand{\s}{\boldsymbol{s}}
\renewcommand{\t}{\boldsymbol{t}}
\newcommand{\p}{\boldsymbol{p}}
\newcommand{\h}{\boldsymbol{h}}
\newcommand{\D}{\mathcal{D}}
\renewcommand{\P}{\mathcal{P}}
\newcommand{\Z}{\mathbb{Z}}
\newcommand{\C}{\mathcal{C}}

\newcommand{\1}{\mathbf{1}}

\newcommand{\R}{E}
\newcommand{\B}{F}

\newcommand{\lrceil}[1]{\left \lceil #1 \right\rceil}

\newtheorem{theorem}{Theorem}
\newtheorem{proposition}{Proposition}
\newtheorem{property}{Property}
\newtheorem{lemma}[theorem]{Lemma}

\newtheorem{remark}{Remark}
\newtheorem{observation}{Observation}
\theoremstyle{definition}
\definecolor{britishracinggreen}{rgb}{0.0, 0.26, 0.15}
\newtheorem{definition}{Definition}
\newcommand{\np}[1]{{\footnotesize  [\textbf{\textcolor{red}{#1}} \textcolor{red!60!black}{--Nikita}]\normalsize}}

\newcommand{\iv}[1]{{\textcolor{green}{Ilya: #1}}}
\title{Optimal Codes Correcting Localized Deletions}
\author{%
  \IEEEauthorblockN{\textbf{Rawad Bitar}\IEEEauthorrefmark{1},
                    \textbf{Serge Kas Hanna}\IEEEauthorrefmark{1},
                    \textbf{Nikita Polyanskii}\IEEEauthorrefmark{1}\IEEEauthorrefmark{2},
                    and \textbf{Ilya Vorobyev}\IEEEauthorrefmark{2}}\\
                    
  \IEEEauthorblockA{\IEEEauthorrefmark{1}%
                     Technical University of Munich,
                    Munich, Germany,
                    \{rawad.bitar, serge.k.hanna, nikita.polianskii\}@tum.de}\\
                    
  \IEEEauthorblockA{\IEEEauthorrefmark{2}%
                     Skolkovo Insitute of Science and Technology,
                    Moscow, Russia,
                    vorobyev.i.v@yandex.ru}

\thanks{This project has received funding from the European Research Council (ERC) under the European Union’s Horizon 2020 research and innovation programme (grant agreement No. 801434) and from the Technical University of Munich - Institute for Advanced Studies, funded by the German Excellence Initiative
and European Union Seventh Framework Programme under Grant Agreement
No. 291763. N. Polyanskii's work was supported by the German Research Foundation (Deutsche Forschungsgemeinschaft, DFG) under Grant No. WA3907/1-1. Ilya Vorobyev was supported in part by the Russian Foundation for Basic Research through grant no.~\mbox{20-01-00559}.
}
}
\date{\today}

\begin{document}

\maketitle
\begin{abstract}
We consider the problem of constructing codes that can correct deletions that are localized within a certain part of the codeword that is unknown a priori. Namely, the model that we study is when at most $k$ deletions occur in a window of size $k$, where the positions of the deletions within this window are not necessarily consecutive. Localized deletions are thus a generalization of burst deletions that occur in consecutive positions. We present novel explicit codes that are efficiently encodable and decodable and can correct up to $k$ localized deletions. Furthermore, these codes have $\log n+\mathcal{O}(k \log^2 (k\log n))$ redundancy, where $n$ is the length of the information message, which is asymptotically optimal in $n$ for $k=o(\log n/(\log \log n)^2)$.
\end{abstract}
\section{Introduction}
Localized deletions correspond to a specific class of errors in which symbols are deleted in certain parts of a string that are unknown a priori. In this setting, the deletions are localized in a certain window in the string but these deletions do not necessarily occur in consecutive positions. Therefore, localized deletions are a generalization of burst deletions which were studied in~\cite{L67,B94,C14,Sch17,A20}. Localized and burst deletions are experienced in various applications such as DNA-based storage and file synchronization.

The study of constructing codes that can correct deletions goes back to the 1960s. In 1966, Levenshtein~\cite{L66} showed that any code that can correct $k$ deletions can also correct $k$ insertions, and vice-versa. Moreover, he showed that such codes can also correct any combination of at most $k$ insertions and deletions. Levenshtein also derived fundamental limits which show that the optimal number of redundant bits needed to correct $k$ deletions that are arbitrarily located in a binary string of length $N$ is $\Theta(k\log(N/k))$~\cite{L66}. In addition, he showed that the codes constructed by Varshamov and Tenengolts (VT codes)~\cite{VT65} are capable of correcting a single deletion and have asymptotically optimal redundancy\footnote{We say that the redundancy of a code construction is asymptotically optimal if the ratio between the redundancy of the construction and some lower bound, e.g., sphere-packing bound, on the redundancy approaches $1$ as the code length goes to infinity.}.  Several previous works studied the classical problem of constructing binary codes that correct $k>1$ deletions that are arbitrarily located in a string~\cite{H02,B16,G19,Vardy,GC,H19,Chen18,SimaIT,SimaSYS}. For the classical setting, the state-of-the-art results in \cite{Chen18,SimaIT,SimaSYS} give codes with $\mathcal{O}(k\log N)$ redundancy. The results in \cite{SimaIT} and \cite{SimaSYS} apply for constant $k$, while the result in \cite{Chen18} applies for $k\le N^{1-\alpha}$ with $0<\alpha<1$. Furthermore, the code presented in \cite{SimaSYS} is systematic, whereas the codes in \cite{Chen18} and \cite{SimaIT} are non-systematic.

The document exchange problem is a related problem that can also provide useful insights for correcting deletions and insertions. In the document exchange setting, a sender (Alice) holds a string $x$ and a receiver (Bob) holds a string $y$, where the edit distance between $x$ and $y$ is assumed to be bounded by $k$. The goal is to design an efficient document exchange protocol by which Alice can send Bob a small sketch based on its string $x$, which allows Bob to recover $x$. The state-of-the-art results on document exchange in \cite{H19} and \cite{Chen18} give deterministic protocols with sketch size $\mathcal{O}(k\log^2(N/k))$ for any $k<N$. 

A separate line of work has focused on the problem of correcting deletions that occur in a {\em single} burst, i.e., bits are deleted in consecutive positions. For this setting, Levenshtein showed that at least $\log N +k-1$ redundant bits are required to correct  $k$ deletions that occur in consecutive positions~\cite{L67}. In~\cite{L67}, Levenshtein constructed asymptotically optimal codes that can correct a burst of {\em at most} two deletions. \mbox{Cheng {\em et al.}~\cite{C14}} provided three constructions of codes that can correct a burst of {\em exactly} $k>2$ deletions. The lowest redundancy achieved by the codes in~\cite{C14} is $k\log(N/k+1)$. The fact that the number of deletions in the burst is exactly~$k$, as opposed to at most $k$, is a crucial factor in the code constructions of~\cite{C14}. Schoeny {\em et al.}~\cite{Sch17} proved the existence of codes that can correct a burst of exactly $k$ deletions and have at most \mbox{$\log N+(k-1)\log\log N+k-1$} redundancy, for sufficiently large $N$. The authors of~\cite{Sch17} also constructed codes that can correct a single burst of at most $k$ deletions. The redundancy for the latter case is at most \mbox{$(k-1)\log N+\big(\binom{k}{2}-1\big)\log\log N+\binom{k}{2}+\log\log k$} which improves on a previous result by Bours in~\cite{B94}. The construction from \cite{Sch17} has been further improved to a redundancy of $\lceil\log k \rceil\log n +\left(k(k+1)/2-1\right) \log \log n +c_k''$ in~\cite{gabrys2017codes}. Note that in the aforementioned results \cite{C14,Sch17,B94, gabrys2017codes}, the size of the burst~$k$ is assumed to be a constant, i.e., $k$ does not grow with $N$. Recently, the authors of~\cite{A20} constructed optimal codes for the case of correcting a single burst of at most $k=o((\log N)^{1/3})$ deletions. The redundancy of these codes is $\log N+\binom{k+1}{2}\log\log N + c_k$ for some constant $c_k$ that only depends on $k$.

In this paper, we study a more general model, in which at most $k$ deletions are {\em localized} in a single window of size $k$, where the position of this window is unknown a priori. Codes for correcting localized errors have been previously studied for other types of errors such as substitutions~\cite{lin2001error, elspas1962note}, where the term ``burst-error-correcting code" is commonly used. In our work, we consider a model where deletions are restricted to a window of size $k$, but do not necessarily occur at consecutive positions within this window. This model for deletions was first studied in~\cite{Sch17} under the name of ``bursts of non-consecutive deletions". The authors of~\cite{Sch17} proved the existence of codes for the particular cases where $k=3$ and $k=4$, with redundancy at most $4\log N+2\log\log N+6$ for $k=3$, and at most $7\log N+2\log\log N+4$ for $k=4$. Localized deletions were also studied in~\cite{GClocalized} for $k=o(n)$, where $n$ is the length of the information message. The authors in~\cite{GClocalized} presented explicit codes that have a non-zero probability of failure which vanishes asymptotically in $n$ for a uniform i.i.d. message. These codes have redundancy at most $4\log n+1$  for $k=o(\log n)$, and at most $4k+1$ for $\Omega(\log n)\leq k\leq o(n)$; furthermore, the encoding/decoding complexity of these codes is $\mathcal{O}(n^2)$.

The main contribution of this paper is constructing an efficiently encodable and decodable {\em zero-error} code that can correct up to $k=\mathcal{O}(n/\log^2 n)$ localized deletions with \mbox{$N-n=\log n + \mathcal{O}(k\log^2 (k\log n))$} redundancy.  It follows from the converse bounds in~\cite{L67,Sch17} that the redundancy of our code is asymptotically optimal in $n$ for $k=o(\log n/ (\log \log n)^2)$. Furthermore, the encoding and decoding complexity is $n \cdot \poly (k\log n)$. A comparison to the state-of-the-art results on codes for consecutive deletions and classical $k$-deletion correcting codes is given in Section~\ref{sec3}.

The rest of the paper is organized as follows. In Section~\ref{sec2} we introduce important definitions and notations that are used throughout the paper. We present our code construction and explain all the components of our code in Section~\ref{sec3}. Section~\ref{sec4} concludes the paper and lists some open problems for future research.

\section{Preliminaries}
\label{sec2}
 For simplicity of presentation, hereafter one-based numbering is used and log $n$ stands for the base-two logarithm of $n$.  A vector is denoted by bold lowercase letters, such as $\x$,
and the $i$th entry of the vector $\x$ is referred to as $x_i$. The length of a vector $\x$ is denoted by $|\x|$. The set of integers from $i$ to $j-1$, $1\le i< j$, is abbreviated by $[i,j)$ or $[i, j-1]$.
Given a set of indices $I$ and a vector $\x$, we define the vector $\x_I$ of length $|I|$ to be the restriction of $\x$ to coordinates from $I$. For a string $\x\in\{0,1\}^n$, by $B_{k}(\x)$ denote the set of all possible strings that can be obtained from $\x$ by deleting bits indexed by $I=\{i_1,\ldots, i_{k'}\}$ with $1\le k'\le k$ such that the difference between the maximal element and the minimal element of $I$ is at most $k-1$. Further, we call $B_{k}(\x)$ the $k$-\textit{ball} centered in $\x$. For an integer $\ell$, we define the vector $\mathbf{1}^{\ell}$ to be the all one vector of length $\ell$. We define $\mathbf{0}^\ell$ similarly. We call the vector $\mathbf{1}^r$ as a \emph{run} of $1$'s of length $r$. We say a vector $\x \in \{0,1\}^n$ contains a run of $1$'s of length $r$ if there exists $i\in [1,n-r+1]$ such that $(x_i,x_{i+1},\dots,x_{r+i-1}) = \mathbf{1}^r$. 
\begin{definition}
    We say that a code $\mathcal{C}\subseteq\{0,1\}^n$ corrects  deletions localized in a window of length $k$ if  for any two distinct codewords $\x_1,\x_2\in\C$, the corresponding $k$-balls have the empty intersection, i.e., $B_{k}(\x_1)\cap B_{k}(\x_2)=\emptyset$.
\end{definition}

	Given a string $\a\in \Z^n$, we define the Varshamov-Tenengolts (VT) check sum as follows
	$$
	\mathsf{VT}(\a):=\sum_{i=1}^{n}ia_i.
	$$
\begin{definition}
	Let $\mathcal{P} \subset \{0,1\}^m$ be a family of strings of length $m>0$. For a positive integer $\Delta$, $m<\Delta\le n$, we say that a binary string $\x$ of length $n$ is $(\mathcal{P},\Delta)$-dense, if each interval of length $\Delta$ in $\x$ contains at least one string from $\mathcal{P}$, i.e., for each $i\in [1,n-\Delta+1]$ there exist $j \in [i,i+\Delta-m]$ and $\p\in\mathcal{P}$ such that $\p = \x_{[j,j+m)}$. 
\end{definition}

\begin{definition}\label{def::indicator vector}
	Let $\x$ and $\mathcal{P}$ be a binary string of length $n$ and a family of strings of length $m$, respectively. Then, we define the indicator vector $\mathbbm{1}_{\mathcal{P}}(\x)$ of the family $\mathcal{P}$ in $\x$ to be a vector of length $n$ with entries
	$$ \mathbbm{1}_{\mathcal{P}}(\x)_i: = 
	\begin{cases}
	1, \quad \text{if } \x_{[i,i+m)} = \p \text{ for some }\p\in\mathcal{P} \text{ and } i\leq n-m+1, \\
	0, \quad \text{otherwise}
	\end{cases}
	$$
	Further let $n_{\mathcal{P}}(\x)$ be the number of ones in $\mathbbm{1}_{\mathcal{P}}(\x)$. We define $\a_{\mathcal{P}}(\x)$ to be a vector of length $n_{\mathcal{P}}(\x)+1$ whose $i$-th entry is the distance between positions of the $i$-th and $(i+1)$-st $1$ in the string $(1,\mathbbm{1}_{\mathcal{P}}(\x),1)$.
\end{definition}

\begin{definition}
\label{def:family_of_strings}
 We consider $\P\subset\{0,1\}^{m}$ with $m=k +\lrceil{\log k}+4$ to be the family of strings that end with a run of $1$'s of length $\lrceil{\log k} + 4$ and do not have any run of $1$'s of length $\lrceil{\log k} + 4$ in its first $k$ entries. We refer to the string $\mathbf{1}^{\lrceil{\log k}+4}$ as the \emph{marker} and to any string $\p\in\mathcal{P}$ as the \emph{pattern}.
\end{definition}

\begin{property}[Marker-rich strings]\label{prop:pattern_rich}
   A string $\x$ is said to be marker rich if it contains at least one marker $\mathbf{1}^{\lrceil{\log k} + 4}$ per window of length $$B := (\lrceil{\log k}+4)2^{\lrceil{\log k} + 8}\lceil \log n \rceil.$$
\end{property}
\begin{property}[String with distanced markers]\label{prop:marker_sparse}
  A string $\x$ is said to have distanced markers if every substring of $\x$ of length $$R := (k+\lrceil{\log k }+ 4)(\lrceil{\log n} +\lrceil{\log k} + 8)$$ contains at least one substring of length $m = k+\lrceil{\log k }+ 4$ that does not contain any copy of the marker $\mathbf{1}^{\lrceil{\log k}+4}$.
\end{property}
In Proposition~\ref{pr::properties lead to density}, we shall show that any string $\x$ satisfying Property~\ref{prop:pattern_rich} and Property~\ref{prop:marker_sparse} is a $(\P,\Delta)$-dense string with $\Delta: = R+B-m$.

\section{Code construction}
\label{sec3}
 Our main result is summarized in the following statement.
 \begin{theorem}
 For  $k=\mathcal{O}(n/\log^2 n)$, there exists an efficiently encodable and decodable code that maps an arbitrary string of length $n$ into a string of length $N$ and is capable of correcting deletions localized in a window of length $k$. The encoding and decoding
complexity of the proposed code is
$n \cdot \poly (k\log n)$. Furthermore, for $n\to\infty$ the redundancy of this code is
 $$
 N-n= \log n + 
\mathcal{O}(k \log^2 (k\log n)).
 $$
 \end{theorem}
 We have several remarks and comments illustrating the contribution of our paper. 
 \begin{itemize}[leftmargin=*]
     \item  From the converse bounds~\cite{L67,Sch17}, it follows that the redundancy of such a code, written as $N-n$, has to be at least $\log N + k - 1$. This implies that our construction is asymptotically optimal when $k=o(\log n/ (\log \log n)^2)$, e.g., it holds if $k$ is an absolute constant or $k=\sqrt{\log n}$. 
     \item For the case of consecutive deletions, our construction is asymptotically optimal as the one from~\cite{A20}, but better in terms of complexity and limitations of $k$, namely: $k=o((\log n)^{1/3})$ in~\cite{A20} and $k=o(\log n/ (\log \log n)^2)$ in the proposed code. 
     \item We observe that a $k$-deletion correcting code can be naively used for correcting deletions localized in a window of length $k$. For the regime of 
     sub-polynomial $k=2^{o(\sqrt{\log n})}$,  the best known efficient $k$-deletion correcting code~\cite{H19,SimaSYS} has redundancy $\mathcal{O}(k \log n)$ and a $\poly(n)$ time encoding and decoding algorithms, which are worse than the redundancy and complexities of the suggested construction. 
 \end{itemize}

  Our construction consists of four components. Specifically, we show how to encode a string $\u\in\{0,1\}^n$ into a string $\widetilde \x\in\{0,1\}^N$ with $N = n + \log n (1+o(1)) + 
 \mathcal{O}(k \log^2 (k\log n))$, where 
$$
\widetilde \x := (\x, H_{\text{sep}}, H_{\text{loc}}(\x), H_{\text{cor}}(\x)),\quad \x := E(\u). 
$$
Suppose that $\widetilde \y$ is the channel output if $\widetilde\x$ is transmitted over the channel, i.e., $\widetilde \y\in B_{k}(\widetilde \x)$. 
Let $k':=|\widetilde \x|-|\widetilde \y|$ be the actual number of deletions. Based on $\widetilde\y$, we will show that it is possible to reconstruct $\u$. 
The injective map $E(\u):\{0,1\}^n\to \{0,1\}^{n+k + 3\lrceil{\log k} + 12}$ that encodes a string $\u$ into a $(\P,\Delta)$ string $\x=E(\u)$ will be discussed in Section~\ref{sec::encoding into dense strings}. 
The vector $H_{\text{sep}}\in\{0,1\}^{k +1}$ makes it possible to deduce whether at least one deletion occurs in the part $\x$.
 If at least one deletion occurs in the part $\x$, then the part  $(H_{\text{loc}}(\x), H_{\text{cor}}(\x))$ is error-free. More details on $H_{\text{sep}}$ will be given in Section~\ref{sec::separation}. The hash function $H_{\text{loc}}(\x)$ that requires  $\log n  +\mathcal{O}(1)$ bits and allows us to locate deletions in the erroneous part $\x$  up to an interval of length $\mathcal{O}(\Delta^2)$ will be presented in Section~\ref{sec::locating deletions}. 
Finally, in Section~\ref{sec::correcting deletions}, we describe  the function $H_{\text{cor}}(\x)$  correcting deletions occurred in  $\x$. This part requires a small amount of redundancy $\mathcal{O}(k \log^2 (k\log n))$ since the location of deletions is known up to a small interval. Note that based on $\x=E(\u)$, we can extract the information string $\u$. 

\subsection{Encoding into $(\P,\Delta)$-dense strings}\label{sec::encoding into dense strings}
We explain the injective map $E(\u):\{0,1\}^n\to \{0,1\}^{n+ k + 3\lrceil{\log k} + 12}$  that encodes $\mathbf{u}\in \{0,1\}^n$ into a $(\mathcal{P},\Delta)$-dense string $\x$. The encoding follows the same procedures from~\cite{SimaIT} and is replicated here for completeness. Let $\P\subset\{0,1\}^{m}$ to be the family of strings as defined in Definition~\ref{def:family_of_strings}. Before we explain the construction we prove the following proposition.

\begin{proposition}\label{pr::properties lead to density}
A string $\x \in \{0,1\}^n$ that satisfies Property~\ref{prop:pattern_rich} and Property~\ref{prop:marker_sparse} is $(\P,\Delta=R+B-m)$-dense.
\end{proposition}

\begin{proof}
Denote by $t_1,<t_2\dots<t_J$ the locations where $\mathbbm{1}_{\mathcal{P}}(\x)$ is equal to $1$. Define $t_0 = 0$ and $t_{J+1} = n+1$. We need to show that for all $i\in [0,J]$, the following holds $$t_{i+1}-t_i\leq B+R-m+1 = \Delta +1.$$

Since $\x$ satisfies Property~\ref{prop:marker_sparse}, then there exists an index $j^\star \in [t_i,t_i+R-k -\lrceil{\log k} - 3]$ such that for every $\ell \in [j^\star, j^\star + k - 1]$ we have $(x_\ell,\dots,x_{\ell + \lrceil{\log k} + 3}) \neq \mathbf{1}^{\lrceil{\log k} + 4}$. Following Property~\ref{prop:pattern_rich}, there exists an integer $a \in [j^\star+1,j^\star+B]$ such that $(x_a,\dots,x_{a+\lrceil{\log k}+3}) = \mathbf{1}^{\lrceil{\log k} + 4}$. Let $w = \min\{a\geq j^\star : (x_a,\dots,x_{a+\lrceil{\log k} + 3}) = \mathbf{1}^{\lrceil{\log k} + 4}\}$. We now have that $w\neq j^\star$ and $w\leq a \leq j^\star + B$. Therefore, $w\in [j^\star+1,j^\star+B]$. In addition, for every $\ell \in [j^\star,w-1]$ we have $(x_{\ell},\dots,x_{\ell+\lrceil{\log k}+ 3})\neq \mathbf{1}^{\lrceil{\log k}+4}$. By definition of $j^\star$ and $w$, we have that $w-j^\star \geq k$. Since $(x_w,\dots,x_{w+\lrceil{\log k}+ 3}) = \mathbf{1}^{\lrceil{\log k }+4}$, we know that $\mathbbm{1}_{\mathcal{P}}(\x)_{w} = 1$. We can now write%
\begin{align*}
    t_{i+1} - t_i &\leq w-t_i\\
    ~& \leq j^\star + B - t_i\\
    ~& \leq R+B-m+1\\
    ~& = \Delta + 1.
\end{align*}
\end{proof}

We are now ready to explain the construction. We split the explanation into three parts. In the first part, we give an encoding that ensures that the string $\x$ is marker-rich, i.e., satisfies Property~\ref{prop:pattern_rich}. In the second part we give an intermediate encoding that would be needed to distance the markers. In the third part we give an encoder that ensures that the markers in the string $\x$ are distanced so that the string satisfies Property~\ref{prop:marker_sparse}. We show at the end that the resulting string $\x$ satisfies both properties and is thus a $(\P,\Delta = R+B-m)$-dense string.

\paragraph{Marker enriching} This stage consists of identifying windows of size $B$ in the input string $\u$ that do not contain a marker. Every such window of bits (substring) is deleted from $\u$. Then, a compressed version of the deleted substring is appended to the end of the input string together with a marker. We first show that all substrings of length $B$ that do not contain a marker can be compressed.

\begin{observation}\label{obs:compress}
Let $\mathcal{S}$ be the set of strings of length $B$ that do not contain a run of $1$'s of length greater than or equal to $\lrceil{\log k}+4$. Then, every string $\mathbf{s} \in \mathcal{S}$ can be compressed to a string $\widetilde{\mathbf{s}} := \phi(\mathbf{s}) \in \{0,1\}^{B-\lrceil{\log n} - 2\lrceil{\log k} -10}$. Here $\phi$ is an invertible map such that $\phi$ and $\phi^{-1}$ can be computed in $O(B^2)$ time.
\end{observation}
The compression works as follows. Divide every string into $2^{\lceil \log k\rceil + 8}\lceil \log n \rceil$ strings each of length $\lceil\log k\rceil+4$. Since every substring cannot be the all one vector, then it can be represented using a symbol from an alphabet of size $2^{\lrceil{\log k}+4}-1$. Therefore, the string $\mathbf{s}$ can be represented by a string $\mathbf{v}$ consisting of $2^{\lceil \log k\rceil + 8}\lceil \log n \rceil$ symbols each from an alphabet of size $2^{\lrceil{\log k}+4}-1$. The number of bits $n_v$ required to represent $\mathbf{v}$ using a binary sequence is given by
\begin{align*}
    n_v & = \lrceil{\log \left( 2^{\lrceil{\log k}+4} -1 \right)^{2^{\lrceil{\log k} + 8} \lrceil{\log n}} }\\
    & = \lrceil{\log \left(1 - 2^{-\lrceil{\log k}-4} \right)^{2^{\lrceil{\log k} + 8} \lrceil{\log n}} } + \left( (\lrceil{\log k}+4) 2^{\lrceil{\log k}+8} \lrceil{\log n}\right)\\
    & \stackrel{(a)}{\leq} 16 \lrceil{\log n}\log\left(\frac{1}{e}\right) + 1 + B\\
    & \leq B - 23 \lrceil{\log n} + 1\\
    & \stackrel{(b)}{\leq} B - \lrceil{\log n} - 2 \lrceil{\log k} -10.
\end{align*}

Here, (a) follows from the fact that for $x>1$, the function $\left(1-\frac{1}{x}\right)^x$ is increasing in $x$ and $\lim_{x\to \infty}  \left(1-\frac{1}{x}\right)^x = \frac{1}{e}$. And (b) holds because for large values of $n$ the following holds $$23 \lrceil{\log n}\geq \lrceil{\log n} + 2 \lrceil{\log k}+11.$$

\emph{Encoding:} We now present the details of the marker enriching process $T_e:\{0,1\}^n \to \{0,1\}^{n+2\lrceil{\log k}+8}$. For a string $\u\in \{0,1\}^n$ let $T_e(\u)$ be the result of the encoding applied on $\u$ to make it rich with markers. The encoding works as follows. First we set $T_e(\u) = \u$ and use a dummy variable $n' = n$. Then we append two markers, i.e., $\mathbf{1}^{2\lrceil{\log k}+8}$, at the end of $T_e(\u)$. We now check the string $T_e(\u)$ for substrings of length $B$ that do not contain a marker. In other words, we look for an integer $i\in [1,n']$ such that for every $j\in [i,i+B-\lrceil{\log k} -4]$ it holds that $$(T_e(\u)_j,T_e(\u)_{j+1},\dots, T_e(\u)_{j+\lrceil{\log k}+3})\neq \mathbf{1}^{\lrceil{\log k}+4}.$$

If we find such an integer $i$ such that $i\leq n'-B+1$, then we delete the substring $(T_e(\u)_i,\dots, T_e(\u)_{i+B-1})$ from $T_e(\u)$ and append $(i,\phi((T_e(\u)_i,\dots, T_e(\u)_{i+B-1}), 0,\mathbf{1}^{2\lrceil{\log k}+8},0)$ to the end of $T_e(\u)$. The appended index $i$ is encoded using a $\lrceil{\log n}$ binary string. We set $n' = n'-B$.

If we find such an integer $i$ such that $i> n'-B+1$, then we delete the substring $(T_e(\u)_i,\dots, T_e(\u)_{n'})$ from $T_e(\u)$ and append $(i,\phi((T_e(\u)_i,\dots, T_e(\u)_{n'},\mathbf{0}^{i+B-n'-1}), 0,\mathbf{1}^{2\lrceil{\log k} + 8-(i+B-n'-1)},0)$ to the end of $T_e(\u)$. The appended index $i$ is encoded using a $\lrceil{\log n}$ binary string. We set $n' = i-1$. In this case the deleted substring has length $n'-i+1$ so that we do not delete any of the appended bits. The quantity $2{\lrceil{\log k}+8}-(i+B-n'-1)$ is always positive  due to appending a run of $1$'s of length $2{\lrceil{\log k}}+8$ in the initialization step. This guarantees that $i+B-n'-1\leq \lrceil{\log k} +4$, otherwise the considered substring contains $(T_e(\u)_{n'+1},\dots,T_e(\u)_{n'+\lrceil{\log k}+4})$ which is a run of $1$'s of length $\lrceil{\log k}+4$ and therefore is not deleted by the algorithm.

We keep looking for values of $i$ that satisfy the aforementioned constraints and applying the explained transformations until deleting all substrings of $\u$ that do not contain a marker.

\emph{Redundancy:} Appending the run of $1$'s adds $2\lrceil{\log k}+8$ redundant bits. All other operations do not increase the length of the string. Therefore, the overall redundancy is at most $2\lrceil{\log k}+8$ bits. To see this, recall from Observation~\ref{obs:compress} that for any string $\s\in \mathcal{S}$, $\widetilde{\s}=\phi(\s)$ is expressed using $B - \lrceil{\log n} - 2 \lrceil{\log k} -10$ bits. The two appended markers and the two $0$'s take $2 \lrceil{\log k} +10$ bits. The index $i$ is expressed using $\lrceil{\log n}$ bits. When deleting a sequence that starts at position $i\leq n'- B+1$, then the appended sequence has a length exactly $B$. When deleting a sequence that starts at position $i> n'- B+1$, then the appended sequence has a length $B-(i+B-n'-1) \leq B$.

\emph{Correctness:} We show that the encoding ensures that any string of length $B$ in $T_e(\u)$ contains the marker $\mathbf{1}^{\lrceil{\log k}+4}$. First we note that $n'$ is always the split index between the original string and the appended parts. Thus, no bits from the added strings are deleted throughout the process. Having $n'$ be the split index is guaranteed by appending the run of $1$'s in the beginning and only deleting bits with index less than $n'$. In addition, $n'$ is decreased every time a string is deleted. Since $n'$ is always the split index, then all strings before $n'$ satisfy the desired property, otherwise the encoding keeps going. In the appended sequences, if $i\leq n'-B+1$ a marker of desired length is added. If $i> n'-B+1$, then $2\lrceil{\log k}+8-(i+B-n'-1) \geq \lrceil{\log k} + 4$ and a marker of the desired length is added. Moreover, the start index of any two markers in the added strings is at most $B-\lrceil{\log k}-4$ bits apart. Therefore, the string $T_e(\u)$ satisfy the desired properties. 

\emph{Decoding strategy:} The decoding follows a reverse procedure of the encoder. We show that when decoding $T_e(\u)$ the output of the decoder is unique and corresponds to the input string $\u$. The decoder looks at the last bit of the input string. All appended strings in the encoding end with a $0$. Therefore, if the last bit of $T_e(\u)$ is a $0$, there is an appended string. The decoder counts the length of the run of $1$'s before that $0$ to determine the length of the originally deleted string. The length of the deleted string is equal to $B$ if the run is of length $2\lrceil{\log k}+8$. Otherwise, the length of the deleted string can be determined using $2\lrceil{\log k}+8$ minus the length of the run of $1$'s. The decoder deletes the $0$ and the run of $1$'s. Now the decoder can invert the compressed version of the deleted string and insert the original string at position $i$ read from the header of the appended sequence. The decoder deletes the remaining part of this appended string. The decoder repeats this procedure until the last bit of the modified string is $1$. At that point the decoder has reached the run of $1$'s of length $2\lrceil{\log k}+8$ appended in the first step of the encoding. The decoder deletes this run and outputs $\u$. The output is guaranteed to be unique because all appended sequences are encoded using the bijective function $\phi$. The recovered indices of the deleted sequences is computed given the current value of $n$ and is guaranteed to be correct. The decoder ends if and only if the last bit of the string is $1$ which is designed by the encoder to be the correct stopping criteria.

\emph{Complexity:} We claim that the total complexity of the encoding process can be reduced to $\mathcal{O}(nB)$. Let a string $\s$ be equal to $(\u, \mathbf{1}^{2\lrceil{\log k}+8})$. We write the encoded sequence into strings $\t_1$ and $\t_2$. The string $\t_1$ will store the string $\s$ without the deleted substrings. The string $\t_2$ will store compressed versions of the deleted substrings. The resulting string $\x$ is equal to the concatenation of $\t_1$ and $\t_2$. 

Scan the string $\s$ from left to right, copy its symbols into $\t_1$ and count the number of consecutive ones in the string $\t_1$. If we have at least $\lrceil{\log k}+4$ consecutive ones in some position, it means that we have found a marker. Save this position to a special variable. Update this position every time a new marker is found.
If in some position the distance to the end of the last marker is bigger than $B$, then we have found a substring in $\s$ of length $B$ without a marker. Cut the last $B$ symbols from the string $\t_1$, compress them using the map $\phi$ and append the result to the string $\t_2$. Continue scanning of the string $s$. 

It is easy to see that the described algorithm implements the encoding process. The complexity of this procedure is $\mathcal{O}(n)$ since each symbol can be scanned, inserted in $\t_1$, used in compression and appended to $\t_2$ at most one time.

To decode the sequence we can find $n'$ and the partition of $\x$ into $\t_1$ and $\t_2$. Then we write the string $\s$ from right to left by scanning the string $t_1$ from right to left and inserting decompressed parts from $\t_2$ in the proper moments. The decoding also has a linear complexity.  

\paragraph{Intermediate encoding} We describe an intermediate step that takes a string of length $k+\lrceil{\log k }+ 4$ that contains some copies of the markers and deletes all the marker. The presented encoding must be invertible. Thus, we present next the encoding and decoding procedure.


\emph{Encoding:} Consider a string $\c \in \{0,1\}^{k+\lrceil{\log k }+ 4}$ that contains at least one copy of the marker. The encoding $T_d:\{0,1\}^{k+\lrceil{\log k }+ 4}\to \{0,1\}^{k+\lrceil{\log k }+ 3}$ explained next deletes all the markers, i.e., encodes $\c$ into $T_d(\c)\in\{0,1\}^{k+\lrceil{\log k }+ 3}$ such that $T_d(\c)$ does not contain any copy of the marker. The encoding is invertible, i.e., given the encoding $T_d(\c)$ one can recover $\c$.

Similarly to the first presented encoding, $T_d$ consists of a series of deleting and appending strings. First, the encoders assigns $T_d(\c) = \c$. Then, the encoder appends a $0$ to the end of $T_d(\c)$. Now, the encoder finds the smallest value of $i\in [1, k]$ such that $T_d(\c)_i = \dots = T_d(\c)_{i+\lrceil{\log k}+3} = 1$. The encoder deletes $(T_d(\c)_i,\dots, T_d(\c)_{i+\lrceil{\log k+3}})$ from $T_d(\c)$ and appends $(i,0,0)$ to the end of $T_d(\c)$. The value of $i$ is encoded by $\lrceil{\log k}$ bits. The encoder sets $n' = k$. We refer to this as the initialization step.

Subsequently, the encoder looks for the smallest value of $i\leq n'$ such that $T_d(\c)_i = \dots = T_d(\c)_{i+\lrceil{\log k}+3} = 1$. If such an integer $i$ exists, the encoder deletes $(T_d(\c)_i,\dots, T_d(\c)_{i+\lrceil{\log k}+3})$ from $T_d(\c)$ and appends $(i,0,0,0,1)$ to the end of $T_d(\c)$. The encoder changes the value of $n'$ to $n' = n' -\lrceil{\log k} - 4$. The encoder keeps repeating this process until $n' \leq \lrceil{\log k} +3$ or no values of $i\leq n'$ exist for which $i$ is the start of a marker.

This encoding procedure indeed results in a string that is one bit shorter than the original string. The initialization step appends a $0$. Then, it deletes a string and appends another one that is $2$ bits shorter. All other deleting and appending processes delete and append a string of the same length. One can verify that the appended strings are not modified by the encoding due to decreasing $n'$ and the appended $0$ at the initialization step. In addition, the resulting string does not contain any copy of the marker. A string that contains the marker is always deleted. All appended strings are guaranteed to have the bit $\lrceil{\log k}+1$ equal to $0$ which ensures that no marker exists.

\emph{Decoding of $T_d(\c)$:} The decoding is straightforward. The appended string in the initialization step ends with a $0$. All other appended strings end with a $1$. The decoder looks at the last bit of the string input string $T_d(\c)$. If the last bit is a $1$, it skips the $1$ and the appended $0$'s, and gets the integer value of $i$ from the $\lrceil{\log k}$ bits representing $i$. The decoder deletes the appended string and adds a marker at position $i$. This process is repeated until the last bit is $0$. In this case, the decoder skips the appended $0$'s and gets the integer value of $i$ from the $\lrceil{\log k}$ bits representing $i$. The decoder deletes the appended string and adds a marker at position $i$. 
The encoding and decoding processes can be implemented effectively in the same manner as in the marker enriching step. The total complexities of decoding and encoding are linear, i.e. $O( k)$. 

\paragraph{Encoding into $(\P,\Delta=R+B-m)$-dense strings} We present the encoder $$T: \{0,1\}^{n+2\lrceil{\log k} + 8}\to \{0,1\}^{n+ k +3\lrceil{\log k} + 12}$$ that takes $T_e(\u)$ as input and outputs $T(\u)$ that is $(\P,\Delta)$-dense. To do so, the encoding ensures that every substring of $T(\u)$ of length $$R = (k+\lrceil{\log k}+4)(\lrceil{\log n} +\lrceil{\log  k} + 8)$$ contains at least one substring of length $m = k + \lrceil{\log k} + 4$ that does not contain any copy of the marker $\mathbf{1}^{\lrceil{\log k}+4}$. 

The encoding is as follows. Let $T(\u) = T_e(\u)$. The encoder appends $(\mathbf{0}^{k}, \mathbf{1}^{\lrceil{\log k}+4})$ to the end of $T(\u)$. Let $n'=n+2\lrceil{\log k}+8$ (the length of $T_e(\u)$). The encoder searches for an integer $i\leq \min \{n',n+k + 3\lrceil{\log k} + 13 - R \}$ such that for every $j \in [i,i+R-k - \lrceil{\log k}-3]$, there exists an integer $\ell \in [j,j+k - 1]$ satisfying $(T(\u)_\ell,\dots,T(\u)_{\ell +\lrceil{\log k} +3}) = \mathbf{1}^{\lrceil{\log k} + 4}$. If such an integer exists, the encoder splits $(T(\u)_i,\dots,T(\u)_{i+R-1})$ into $(\lrceil{\log n} +\lrceil{\log k} + 8)$ blocks $\b_1, \dots, \b_{\lrceil{\log n} +\lrceil{\log k} + 8}$ of length $k+\lrceil{\log k} + 4$ each. The encoder deletes $\b_2, \dots, \b_{\lrceil{\log n} +\lrceil{\log k} + 7}$ from $T(\u)$ and appends $$(0,T_d(\b_2), \dots, T_d(\b_{\lrceil{\log n} +\lrceil{\log k} + 7}),i+k + \lrceil{\log k}+3, \mathbf{1}^{\lrceil{\log k }+4},0)$$ to the end of $T(\u)$. The appended number $i+k + \lrceil{\log k}+3$ is represented with a binary string of length $\lrceil{\log n}$. The encoder sets $n'=n'-R + 2k + 2\lrceil{\log k} + 8$. The encoder repeats the same process until there is no more values of $i$ that satisfy the desired properties.

In terms of redundancy, the encoding only adds $m$ bits. Every deleted string has the same length of the appended string. This is because the intermediate encoding $T_d$ reduces one bit per encoded sequence. We have $\lrceil{\log n} +\lrceil{\log k} + 6$ encoded sequences. The encoder uses $\lrceil{\log n}$ bits to append $i+k + \lrceil{\log k}+3$. The added marker takes $+\lrceil{\log k} + 4$ bits and the remaining $2$ bits are taken by the added $0$'s. Therefore, the only addition made by this encoding is the initial addition of the string $(\mathbf{0}^{k}, \mathbf{1}^{\lrceil{\log k}+4})$.

Similarly to previous encoding algorithms, the appended sequence in the initialization and the decrease of $n'$ at every round (deleting and appending) guarantee that the encoder does not modify the appended strings.

\emph{Correctness:} We now prove that the resulting string $T(\u)$ satisfies Property~\ref{prop:pattern_rich} and Property~\ref{prop:marker_sparse}. We start by showing by induction over the number of rounds $r$ (searching for an index $i$, deleting and appending a sequence) that $T(\u)$ satisfies Property~\ref{prop:pattern_rich}. At the initial step, $r=0$, the string $T(\u) = (T_e(\u),\mathbf{0}^{k}, \mathbf{1}^{\lrceil{\log k}+4})$ satisfies Property~\ref{prop:pattern_rich} by construction of $T_e$. Hence, the hypothesis holds for $r=0$. Assume that after $r>0$ rounds $T(\u)$ satisfies the desired property. At round $r+1$, after a string is deleted, the blocks $\b_1$ and $\b_{k +\lrceil{\log k} + 8}$ are not deleted. Those blocks both contain the marker. Therefore, after the deletion step the string $T(\u)$ still satisfies Property~\ref{prop:pattern_rich}. All appended strings contain the marker $\mathbf{1}^{\lrceil{\log k}+4}$. The index distance between any two markers in the appended sequences is at most $R-2k - 2\lrceil{\log k} - 8 \leq B - \lrceil{\log k} - 4$. Thus, the appended strings satisfy Property~\ref{prop:pattern_rich} and so does the whole string $T(\u)$.

Now we prove that $T(\u)$ satisfies Property~\ref{prop:marker_sparse}. Due to the properties of the encoder, for any value $i\in [1, \min\{n',n+k + 3\lrceil{\log k} + 13 - R \}]$, there exists some $j\in [i,i+R-k - \lrceil{\log k} - 3]$ such that for every $\ell \in [j,j+k - 1]$ the following holds $(T(\u)_\ell,\dots,T(\u)_{\ell +\lrceil{\log k} +3}) \neq \mathbf{1}^{\lrceil{\log k} + 4}$. Recall that all appended strings satisfy Property~\ref{prop:marker_sparse} and are not deleted by the algorithm. 
We only have to prove the property for $i\in [n'+1, n+k + 3\lrceil{\log k} + 13 - R]$. Since the algorithm does not delete the appended strings, the interval $[i,i+R-1]$ contains the string $(0,T_d(\b_2))$ of length $m$ that comes from an appended string. According to the intermediate encoding $T_d$, the string $T_d(\b_2)$ does not contain the marker. Hence, for $i\in [n'+1, n+k + 3\lrceil{\log k} + 13 - R]$ there exists a string of length $m$ that does not contain the marker. Thus, $T(\u)$ satisfies Property~\ref{prop:marker_sparse} as well.

\emph{Decoding of $T(\u)$:} We show that the decoder outputs $T_e(\u)$ from which we have previously shown that $\u$ can be uniquely determined. The decoding procedure follows the reverse steps done by the encoding. The decoder looks at the last bit of $T(\u)$. If this bit is $1$, then only the pattern $(\mathbf{0}^{k}, \mathbf{1}^{\lrceil{\log k}+4})$ has been added and no intermediate encoding has been done. The decoder deletes the additional pattern and obtains $T_e(\u)$. Otherwise, the last string of length $m(\lrceil{\log n}+\lrceil{\log k} + 6)$ is encoded using the intermediate encoding $T_d(\c)$. The decoder deletes this string from the input string, then determines the strings $T_d(\b_2),\dots, T_d(\b_{\lrceil{\log n}+\lrceil{\log k} + 7})$ encoded using $T_d(\c)$ and inverts them. Then, the encoder reads the index $i+m-1$ and inserts $\b_2,\dots,T_d(\b_{\lrceil{\log n}+\lrceil{\log k} + 7}$ at this position. This step is repeated until the last bit of the input string is $1$.

\emph{Complexity:} The encoding and decoding can be implemented effectively in linear time $\mathcal{O}(n)$ in the similar manner as in the marker enrichment step. Indeed, let a string $\s$ be equal to $(T_e(\u), \mathbf{0}^k, \mathbf{1}^{\lrceil{\log k}+4})$. We write the encoded sequence into strings $\t_1$ and $\t_2$. The string $\t_1$ will store the string $\s$ without the deleted substrings. The string $\t_2$ will store compressed versions of the deleted substrings. The resulting string $\x$ is equal to the concatenation of $\t_1$ and $\t_2$. 

Scan the string $\s$ from left to right, copy its symbols into $\t_1$ and
compute the distance to the last marker and the distance to the last substring of length $m$ without a marker in the similar manner as in the marker enrichment step. If at some moment the distance to the start of the last substring of length $m$ without a marker is bigger than $R$, 
then we delete $R$ last symbols from $\t_1$, append encoding of deleted part to $\t_2$, append the first and the last blocks of the deleted part to $\t_1$, compute the actual distances to the last marker and the last string of length $m$ without a marker. Continue scanning of the string $\s$.
One can see that this algorithm implements the encoding process and works in linear time. 
The decoding can be done in linear time analogously to the decoding in the marker enrichment step.

\subsection{Protecting the hashes}\label{sec::separation}
As previously mentioned, we use the hashes $H_{\text{loc}}(\x)$ and $H_{\text{cor}}(\x)$ to locate and correct potential deletion errors that may occur in $\x$.
Since the deletions could also affect $H_{\text{loc}}(\x)$ and $H_{\text{cor}}(\x)$, we need to protect these hashes and enable their recovery at the decoder in order to be able to correct deletions in $\x$. One way to this is to encode $(H_{\text{loc}}(\x),H_{\text{cor}}(\x))$ using a deletion correcting code. However, given the localized nature of the deletions in our problem, we propose the following simpler solution. 

To protect the hashes, we insert a buffer $H_{\text{sep}}$ of size $k+1$ that separates $(H_{\text{loc}}(\x),H_{\text{cor}}(\x))$ from $\x$. Since the $k'\leq k$ deletions are localized in a window of size at most $k$, a consequence of inserting the buffer of size $k+1$ is that the deletions now cannot affect $\x$ and $(H_{\text{loc}}(\x),H_{\text{cor}}(\x))$ simultaneously. We design $H_{\text{sep}}$ in a way such that we can detect whether the deletions have affected $\x$ or not. Therefore, if we detect that the deletions have affected $\x$, we know that the hashes are error-free so we use them to proceed with our decoding methods explained in the subsequent sections. Otherwise, we detect that no deletions have affected $\x$, then we know that the deletions have only affected $(H_{\text{sep}},H_{\text{loc}}(\x),H_{\text{cor}}(\x))$, so we immediately recover $\x$. Next, we explain how we design this buffer and detect whether the deletions have affected $\x$ or not.

We set the buffer of length $k+1$ to $H_{\text{sep}}\coloneqq (\mathbf{0}^{k},1)$. Consider the transmitted sequence $\widetilde \x$ that is affected by $k'\leq k$ localized deletions resulting in the received sequence $\widetilde \y$ of length $N-k'$.  The buffer $H_{\text{sep}}$ is composed of $k$ zeros followed by a single one, and its position in $\widetilde \x$ ranges from the \mbox{$(|\x|+1)^{st}$}~bit to the \mbox{$(|\x|+k+1)^{st}$}~bit. Let $\widetilde y_{\alpha}$ be the bit in position $\alpha$ in the received string, where \mbox{$\alpha\coloneqq |\x|+k-k'+1$}. 
The decoder observes $\widetilde y_{\alpha}$:
\begin{enumerate}
    \item If \mbox{$\widetilde y_{\alpha}=1$}, then this means that the one in the buffer has shifted $k'$ positions to the left because of the deletions, i.e., all the deletions occurred to the left of the one in the buffer. In this case, we consider that the deletions have affected $\x$ and we detect that $(H_{\text{loc}}(\x),H_{\text{cor}}(\x))$ are error-free.
    \item If $\widetilde y_{\alpha}=0$, then this indicates that the $k'$ deletions occurred to the right of the first zero in the buffer, i.e., $\x$ was unaffected. In this case, the decoder simply outputs the first $|\x|$ bits of the received string. 
\end{enumerate}

\subsection{Locating deletions}\label{sec::locating deletions}
To determine the function $H_{\text{loc}}(\x)$, we use the VT check sum, which is defined over integers, similar to that in \cite{L67,A20}. A key observation of this section is given in the following statement.
\begin{observation}[Localized deletions and $(\P,\Delta)$-dense strings] \label{obs:dense strings and deletions}
Let $\x$ be a $(\P,\Delta)$-dense string with $\P$ being as stated in Definition~\ref{def:family_of_strings}. Then, deletions localized in a window of length $k$ do not destroy nor create more than two patterns from $\P$ in $\x$.
\end{observation}
\begin{proof}
We note that the distance between 1's in the vector $\mathbbm{1}_{\mathcal{P}}(\x)$ is at least $k+1$, i.e., patterns from $\P$ are located in $\x$ at distance at least $k+1$. Therefore, deletions localized in a window of length $k$ occur in at most two patterns in $\x$ and, thus, do not destroy more than two patterns. Similar arguments work when localized deletions result in creating new patterns in $\x$.
\end{proof}
Now we proceed by describing the encoding and decoding algorithm of the locating procedure.

\textit{Encoding:} Let us define $c_1  := n_{\P}(\x) \pmod{5}$, $c_1\in[0,4]$, and $c_2 := \mathsf{VT}(\a_{\P}(\x)) \pmod{6n}$, $c_2\in [0, 6n-1]$. Assume that there is a standard map from non-negative integers at most $n_1$ to their binary representations of length $\lceil\log n_1 \rceil$. Then we define the function $H_{\text{loc}}(\x) := \left(c_1,c_2\right)$.

\textit{Decoding strategy:}
Utilizing the arguments in Section~\ref{sec::separation}, we can assume that $c_1$ and $c_2$ are known before we start the decoding procedure. Define $\y:=\widetilde \y_{[1,|\x|-k']}$ and note that $\y \in B_{k'}(\x)$. 
We want to locate the positions of the localized deletions up to an interval of length $\mathcal{O}(\Delta^2)$.
 Recall that $a_{\P}(\x)_i \leq \Delta$ as $\x$ is $(\P,\Delta)$-dense. By  Observation~\ref{obs:dense strings and deletions}, we obtain that $a_{\P}(\y)_i\leq 3\Delta$ and $n_{\P}(\x)-2\le n_{\P}(\y)\le n_{\P}(\x)+2$.

Compare the vectors $\a_{\P}(\x)$ and $\a_{\P}(\y)$. We claim that the second vector can be obtained from the first by replacing the substring $\s_1$ by the substring $\s_2$, $0 \leq |\s_1|, |\s_2| \leq 3$, $||\s_1|-|\s_2||\leq 2$. Let $\s_1$ and $\s_2$ be equal $\a_{\p}(\x)_{[b,e_1)}$ and $\a_{\p}(\y)_{[b,e_2)}$ with $e_1:=b+|\s_1|$ and $e_2:=b+|\s_2|$. In other words, the vectors $\a_{\P}(\x)$ and $\a_{\P}(\y)$ can be represented as concatenations in the following way
\begin{align}
    \a_{\P}(\x)&=(\a_{\P}(\x)_{[1,b)},\s_1, \a_{\P}(\x)_{[e_1,n_{\P}(\x)]})\nonumber\\ &=(\a_{\P}(\y)_{[1,b)},\s_1, \a_{\P}(\y)_{[e_2,n_{\P}(\y)]}),\nonumber\\
    \a_{\P}(\y)&=(\a_{\P}(\x)_{[1,b)},\s_2, \a_{\P}(\x)_{[e_1,n_{\P}(\x)]})\nonumber\\
    &=(\a_{\P}(\y)_{[1,b)},\s_2, \a_{\P}(\y)_{[e_2,n_{\P}(\y)]}).\label{eq::view of vectors}
\end{align}
Let $d$ be equal to $|\s_2|-|\s_1|$. Observe that the value $d$ can be computed based on $\y$ as we require $c_1=n_{\P}(\x) \pmod{5}$ and, thus, $d=n_{\P}(\y)-c_1 \pmod{5}$. 
Using the property~\eqref{eq::view of vectors} and the fact $\sum_{i=b}^{e_1-1}a_{\P}(\x)_i - \sum_{i=b}^{e_2-1}a_{\P}(\y)_i = k'$, we compute the difference between the VT checksums for vectors $\a_{\P}(\y)$ and $\a_{\P}(\x)$:
\begin{align*}
D&:=	\mathsf{VT}(\a_{\P}(\y))- \mathsf{VT}(\a_{\P}(\x))\\
&= d\sum\limits_{i=e_2}^{n_{\P}(\y)} a_{\P}(\y)_i -\sum\limits_{i=b}^{e_1-1}ia_{\P}(\x)_i+\sum\limits_{i=b}^{e_2-1}ia_{\P}(\y)_i\\
&=d\sum\limits_{i=e_2}^{n_{\P}(\y)} a_{\P}(\y)_i -b\left(\sum\limits_{i=b}^{e_1-1}a_{\P}(\x)_i+\sum\limits_{i=b}^{e_2-1}a_{\P}(\y)_i\right)\\
&\quad -\sum\limits_{i=b+1}^{e_1-1}(i-b)a_{\P}(\x)_i+
\sum\limits_{i=b+1}^{e_2-1}(i-b)a_{\P}(\y)_i
\\
&=d\sum\limits_{i=b+3}^{n_{\P}(\y)} a_{\P}(\y)_i-bk'+ \R,
\end{align*}
where $\R$ is defined by 
$$
\R:=d\sum\limits_{i=e_2}^{b+2} a_{\P}(\y)_i-\sum\limits_{i=b+1}^{e_1-1}(i-b)a_{\P}(\x)_i+
\sum\limits_{i=b+1}^{e_2-1}(i-b)a_{\P}(\y)_i.
$$
We observe that $e_1, e_2\leq b+3$ and $|d|\leq 2$. Therefore, the first and the last summations in the definition of $\R$ 
has at most three addends with the coefficient $\leq \max(|d|, e_2-b-1)$.
 Therefore, $|\R|$ can be upper bounded as $18\Delta$. Finally, we derive that
$$
D=d\sum\limits_{i=b+2}^{n_{\P}(\y)} \a_{\P}(\y)_i-bk'+ \R,\quad |\R|\leq 18\Delta.
$$
Note that $-3n-18\Delta \le D \le 2n + 18\Delta$. Thus, for $n>36\Delta$, we have $6n>5n+36\Delta$ and it suffices to know $c_2=\mathsf{VT}(\a_{\P}(\x)) \pmod{6n}$ to compute $D$. Define the function
$$
G(s):=d\sum\limits_{i=s+2}^{n_{\P}(\y)} \a_{\P}(\y)_i-sk'.
$$
We observe that $D-\R=G(b)$  and the function $G(s)$ is a strictly monotone function of $s$. Indeed, 
if $d\geq 0$ then $G(s)$ is a strictly decreasing function of $s$. 
    If $d<0$ then using the fact $a_{\P}(\y)_i\ge m>k \ge k'$ we conclude that $G(s)$ is a strictly increasing function of $s$. Since it is easy to compute $\R$ based on $\y$, we make use of the bound $|\R|\le 18\Delta$ and  define the set 
    $$
    \B:=\left\{s:\ G(s)\in [D-18\Delta,D+18\Delta]\right\}
    .
    $$
    Note that the difference between two elements from  $\B$ cannot be larger than $36\Delta$. 
Since each $a_{\P}(\y)\leq 3\Delta$, we can locate a segment of length at most $108\Delta^2+3\Delta=\mathcal{O}(\Delta^2)$, where all deletions are located. Note that computing the vector $\a_{\P}(\y)$, finding the set $\B$ and determining the required interval of length $\mathcal{O}(\Delta^2)$  can be done in $\mathcal{O}(n)$ time.

\subsection{Correcting deletions}\label{sec::correcting deletions}

In this step, we make use of the following result from~\cite{Chen18}.

\begin{theorem}[Theorem 1.3, Remark 1.5 from~\cite{Chen18}]\label{th::genetal edit errors}
For any positive integers $n$ and $k$ with $k \leq n/4$, there exists an explicit code which maps any
message $\x$ of length $n$ into a codeword $(\x, h(\x))$ of length $n + \mathcal{O}(k\log^2(n/k))$ and
corrects up to $k$ deletions and insertions. The complexities of encoding, i.e. computing of function $h$, and decoding algorithms are polynomial in $n$. 
\end{theorem}
\begin{remark}
We emphasize that Theorem 1.3 from~\cite{Chen18} allows us to encode in a systematic way, i.e., each codeword contains an original message. This property is important for our further analysis. 
Although for our range of parameters Theorem 1.4 from~\cite{Chen18} gives smaller redundancy, the encoding is not systematic.
We note that codes with the same redundancy as in Theorem~\ref{th::genetal edit errors} were obtained in~\cite{H19}.
\end{remark}
 By the argument in Section~\ref{sec::locating deletions}, we can find the segment of length at most $108\Delta^2+3\Delta$ such that all deletions are situated inside this segment. Let $M:=108\Delta^2+3\Delta$. Recall that $\Delta = \mathcal{O}(k \log k \log n)$ as $n\to\infty$. Now we show how to define $H_{\text{cor}}(\x)$.
 
\textit{Encoding:} Append the smallest amount of zeroes to the end of $\x$ to obtain the vector $\hat{\x}$, which length $n'$ is divisible by $M$. Partition a codeword $\hat{\x}$ into $L=n'/M$ parts of length $M$, $\hat{\x}=(\x_1, \x_2, \ldots, \x_L).$ 
For each message $\x_i$, we use the construction from Theorem~\ref{th::genetal edit errors} to compute parity bits $\h_i=h(\x_i)$ to protect them against $2k$ deletions. Note that the length 
of $\h_i$ is $\mathcal{O}(k\log^2( M/(2k)))=\mathcal{O}(k\log^2k+k(\log\log n)^2)$.
 Define $\h_{even}$ and $\h_{odd}$ to be the following bit-wise sums
\begin{align*}
    \h_{\text{even}}&=\sum\limits_{i=1}^{\lfloor L / 2\rfloor}\h_{2i}, & 
    \h_{\text{odd}}&=\sum\limits_{i=1}^{\lceil L / 2\rceil}\h_{2i-1}.
\end{align*}
Finally, we let $H_{\text{cor}}(\x)$ to be $(\h_{\text{even}}, \h_{\text{odd}})$.

\textit{Decoding strategy:}  Since the position of deletions is already located, we can recover $\x_j$ for all $j$ except possibly two consecutive blocks indexed by $l$ and  $l+1$ for some integer $l$.
Append $LM-|\x|$ zeroes to the end of the  string $\y$ to obtain $\hat\y$. Divide $\hat\y$ into blocks $(\y_1, \y_2, \ldots, \y_L)$ such that $|\y_j|=M$ for $j\ne l$ and $|\y_l|=M-k'$. Clearly, $\x_j=\y_j$ for all $j$ except $j=l, l+1$. For $j=l, l+1$, the string $\y_j$ is obtained from $\x_j$ by at most $k$ deletions and at most $k$ insertions. 
Compute $\h_j=h(\y_j)=h(\x_j)$ for $j\ne l, l+1$. Using $\h_{\text{even}}$ and $\h_{\text{odd}}$, we find $\h_l$ and $\h_{l+1}$. Thus, applying the decoding algorithm from~\cite{H19,Chen18}, we can reconstruct $\x_l$ and $\x_{l+1}$.

\emph{Complexity:} We claim that the encoding complexity of this step is $n\cdot  \poly(k\log n)$. Indeed, the complexity of computing  hash $h_i$ is polynomial in $M$ by Theorem~\ref{th::genetal edit errors}, and the total number of such hashes is at most $n$. The decoding complexity 
is equal to the complexity of computing of hashes and the complexity of decoding of two codes of length $\poly(k\log n)$, that results in $n\cdot  \poly(k\log n)$ total complexity.
\section{Conclusion}
\label{sec4}
In this paper, we propose a novel efficient code  of length $n$ capable of correcting deletions occurred in a window of length $k$, where $k=\mathcal{O}(n/\log^2 n)$. The encoding and decoding algorithms of the proposed construction run in $n\cdot \poly(k\log n)$ time. For $k=o(\log/ (\log \log n)^2)$ and $n\to\infty$, the redundancy of this code is $\log n(1+o(1))$, i.e., it is asymptotically optimal. An interesting question is whether this construction can be generalized for edit errors.
	\bibliographystyle{IEEEtran}
	\bibliography{ref}
\end{document}